\documentclass{article}

\usepackage{hyperref}

\usepackage{utils}

\graphicspath{{figures/}}
\setenumerate{leftmargin=*}
\setitemize{leftmargin=*}
\setdescription{leftmargin=*}

\usepackage[accepted]{icml2023}

\icmltitlerunning{Learning Control-Oriented Dynamical Structure from Data}

\begin{document}

\twocolumn[%
    \icmltitle{Learning Control-Oriented Dynamical Structure from Data}
    \icmlsetsymbol{equal}{*}
    \begin{icmlauthorlist}
        \icmlauthor{Spencer M. Richards}{asl}
        \icmlauthor{Jean-Jacques Slotine}{nsl}
        \icmlauthor{Navid Azizan}{lids}
        \icmlauthor{Marco Pavone}{asl}
    \end{icmlauthorlist}

    \icmlaffiliation{asl}{Autonomous Systems Laboratory (ASL), Stanford University, Stanford, CA 94305, USA}
    \icmlaffiliation{nsl}{Nonlinear Systems Laboratory, Massachusetts Institute of Technology, Cambridge, MA 02139, USA}
    \icmlaffiliation{lids}{Laboratory for Information \& Decision Systems (LIDS), Massachusetts Institute of Technology, Cambridge, MA 02139, USA}
    \icmlcorrespondingauthor{Spencer M. Richards}{spenrich@stanford.edu}
    \icmlkeywords{Control-oriented learning, Learning stabilizing controllers, Stability guarantees, System identification}
    \vskip 0.3in
]

\printAffiliationsAndNotice{}  

\begin{abstract}%
    Even for known nonlinear dynamical systems, feedback controller synthesis is a difficult problem that often requires leveraging the particular structure of the dynamics to induce a stable closed-loop system. For general nonlinear models, including those fit to data, there may not be enough known structure to reliably synthesize a stabilizing feedback controller. In this paper, we discuss a state-dependent nonlinear tracking controller formulation based on a state-dependent Riccati equation for general nonlinear control-affine systems. This formulation depends on a nonlinear factorization of the system of vector fields defining the control-affine dynamics, which always exists under mild smoothness assumptions. We propose a method for learning this factorization from a finite set of data. On a variety of simulated nonlinear dynamical systems, we empirically demonstrate the efficacy of learned versions of this controller in stable trajectory tracking. Alongside our learning method, we evaluate recent ideas in jointly learning a controller and stabilizability certificate for known dynamical systems; we show experimentally that such methods can be frail in comparison.%
    \footnote{We provide code to reproduce all of our results at:\\\url{https://github.com/StanfordASL/Learning-Control-Oriented-Structure}.}
\end{abstract}

\section{Introduction}\label{sec:intro}

Data-driven system identification and control algorithms are imperative to the operation of autonomous systems in complex environments. In particular, model-based algorithms equip an autonomous agent with the ability to learn how it and the system it is part of evolve over time. However, for general nonlinear systems including those learned from data, it is not always clear how to synthesize a \emph{stabilizing} tracking controller. Effective control design often leverages specific system structure; some classic examples of this are the linear quadratic regulator (LQR) for linear dynamics, and the computed torque method and its variants for Lagrangian dynamics \citep{SlotineLi1987,MurrayLiEtAl1994}. A central goal of \emph{control-oriented learning} \citep{RichardsAzizanEtAl2021,RichardsAzizanEtAl2023} and this paper is to jointly learn a dynamics model and additional \emph{control-oriented} structure that naturally encodes or reveals a stabilizing controller design.
\vspace{-1em}

\paragraph{Related Work} An approach favoured by recent works has been to learn stabilizing controllers for nonlinear system models by simultaneously learning a parametric controller and a parametric control-theoretic certificate, such as a control Lyapunov function (CLF) or control contraction metric (CCM). This paradigm originates in works that learn stability certificates for nonlinear systems of the form $\dot{x} = f(x)$ or $x_{t+1} = f(x_t)$. Convergence of the state to $x = 0$ is guaranteed if a Lyapunov certificate function~$V$ can be found such that $\tran{\grad{V}(x)}f(x) < 0$ or $V(f(x)) - V(x) < 0$, respectively, for each~$x \neq 0$. Such functional inequalities serve as the cornerstone for methods that learn parametric certificates from data either via gradient descent on a loss function comprising sampled point violations \citep{RichardsBerkenkampEtAl2018, BoffiTuEtAl2020}, or formal synthesis and verification \citep{AbateAhmedEtAl2021}. Similar functional inequalities appear in contraction theory \citep{LohmillerSlotine1998} to describe the convergence of system trajectories to each other over time, and have been used in imitation learning to regularize fitted dynamics models towards stability \citep{SindhwaniTuEtAl2018} or intrinsic stabilizability \citep{SinghRichardsEtAl2020}.

For \emph{controlled} nonlinear systems like ${\dot{x} = f(x) + B(x)u}$, one can try jointly learning a parametric CLF $V$ and parametric controller $u = k(x)$ by penalizing violations of the inequality $\tran{\grad{V}(x)}\rbr*{f(x) + B(x)k(x)} < 0$ at sampled states. This concept underlies most prior work on learning certified stabilizing nonlinear controllers \citep{ChangRoohiEtAl2019,ChangGao2011,DawsonQinEtAl2021,DawsonGaoEtAl2022}. For tracking a trajectory $(\bar{x}(t),\bar{u}(t))$, \citet{SunJhaEtAl2020} jointly learn a CCM and a feedback controller $u = \pi(x, \bar{x}, \bar{u})$, again based on sampled inequality violations. Such approaches aspire to the closed-loop stability promised by satisfaction of this infinite dimensional constraint, yet it is unclear whether penalizing violations at a finite number of points is sufficient to achieve this in practice.

Rather than trying to fit a controller and certificate to data, one can leverage structure in the dynamics to inform stabilizing controller design. Lagrangian dynamics of the form $H(q)\ddot{q} + C(q,\dot{q})\dot{q} + g(q) = u$ with state $x \defn (q,\dot{q})$ are amenable to feedback linearization \citep{SlotineLi1991} by virtue of their double-integrator form, even when learned from data \citep{GuptaMendaEtAl2020,RichardsAzizanEtAl2021,DjeumouNearyEtAl2022}. Hamiltonian dynamical structure as a physics-based prior in learned models can be exploited to synthesize passivity-based controllers \citep{ZhongDeyEtAl2020,LiDuongEtAl2022}. Perhaps the most fundamental example of structure informing control is LQR, which for linear dynamics $\dot{x} = Ax + Bu$ computes an optimal stabilizing controller from a Riccati equation using the system matrices~$(A,B)$ and chosen cost matrices~$(Q,R)$. Each of these designs is tailored to a subset of control-affine dynamical systems, yet LQR can be extended to general control-affine systems of the form $\dot{x} = f(x) + B(x)u$ with the \emph{state-dependent coefficient (SDC) factorization} ${f(x) = A(x)x}$ \citep{Cloutier1997}, which exists as long as $f$ is continuously differentiable and $f(0) = 0$ \citep{Cimen2010}. A feedback controller can then be implemented by solving the corresponding \emph{state-dependent Riccati equation (SDRE)} in terms of $(A(x), B(x))$ in closed-loop. While such a controller is only locally stabilizing in theory, in practice it has a large region of attraction and has proven effective in automotive \citep{Acarman2009}, spacecraft \citep{CloutierZipfel1999}, robotic \citep{WatanabeIwaseEtAl2008}, and process control \citep{BanksBeelerEtAl2002}.
\vspace{-1em}

\paragraph{Contributions} In this work, we study how to jointly identify nonlinear dynamics models and control-oriented structures from data that can be naturally leveraged in stabilizing closed-loop tracking control design. To this end, we study tracking controller for general nonlinear control-affine systems based on SDRE feedback. While SDREs have seen use in fixed-point stabilization, we focus on its extension to exactly characterizing and controlling the error dynamics for \emph{trajectory tracking}. This extension relies on a generalized SDC factorization of the error dynamics that always exists for continuously differentiable dynamics. We propose a method to learn such structure from a finite data set, and thereby enable the use of SDRE-based tracking control. We compare our method of learning control-enabling structure to an adaptation of prior work that tries to jointly learn a dynamics model, controller, and stability certificate. In a variety of simulated nonlinear systems, we demonstrate that our learned controller performs well in closed-loop, and that controllers instead learned alongside dynamics models and parametric certificate functions can be brittle and data inefficient in practice.

\section{Problem Statement}\label{sec:problem}
\vspace{-0.4em}
In this paper, we are interested in learning to control the nonlinear control-affine dynamical system
\begin{equation}\label{eq:caf}
    \dot{x} = f(x) + B(x)u = f(x) + \sum_{j=1}^m u_j b_j(x),
\end{equation}
with state $x(t) \in \R^n$, control $u(t) \in \R^m$, drift ${f : \R^n \to \R^n}$, and actuator $B : \R^n \to \R^{n \x m}$ with columns $b_j : \R^n \to \R^n,\ j \in \{1, 2, \dots, m\}$. In particular, we want to determine a tracking controller of the form $u = \pi(x, \bar{x}(t), \bar{u}(t))$ such that $(x(t), u(t))$ converges to any \emph{dynamically feasible} pair $(\bar{x}(t), \bar{u}(t))$, i.e., satisfying $\dot{\bar{x}} = f(\bar{x}) + B(\bar{x})\bar{u}$. While we know the dynamics take the form of \cref{eq:caf}, the vector fields $(f, \{b_j\}_{j=1}^m)$ are otherwise \emph{unknown} to us. Instead, we only have access to a finite pre-collected data set $\mathcal{D} \defn \{(x^{(i)}, u^{(i)}, \dot{x}^{(i)})\}_{i=1}^N$ of input-output measurements of \cref{eq:caf}.

\section{Nonlinear Tracking Control}\label{sec:method}
\vspace{-0.4em}
In this section, we overview a number of methods for synthesizing a tracking controller $u = \pi(x, \bar{x}(t), \bar{u}(t))$ for any control-affine nonlinear system of the form in \cref{eq:caf}. We begin with LQR-based methods, including state-dependent-LQR tracking control. We also discuss tracking controllers that are guaranteed to exponentially stabilize the resulting closed-loop dynamics provided an accompanying certificate function is found, namely a control contraction metric (CCM). For each controller, we highlight the \emph{control-oriented} structure that is required in addition to the dynamics to enable a stabilizing feedback signal. We will then discuss how to jointly learn such structure along with a dynamics model from data in \cref{sec:learning-structure} to enable closed-loop tracking control.

\subsection{Linearized LQR}\label{sec:lqr-linear}
Perhaps the simplest approach to tracking control is based on linearizing the dynamics in \cref{eq:caf} around the current target $(\bar{x}(t),\bar{u}(t))$. Specifically, in this method we first linearize the nonlinear dynamics of the tracking error $e(t) \defn x(t) - \bar{x}(t)$ given by
\begin{equation}\label{eq:error}
    \dot{e} = f(x) + B(x)u - f(\bar{x}) - B(\bar{x})\bar{u}
\end{equation}
to arrive at the approximation
\begin{equation}\label{eq:error-linearized}
    \dot{e} \approx \underbrace{\biggl(\pd{f}{x}(\bar{x}) + \sum_{j=1}^m \bar{u}_j \pd{b_j}{x}(\bar{x})\biggr)}_{\eqqcolon A(\bar{x},\bar{u})}e + B(\bar{x})(u - \bar{u}).
\end{equation}
Then, with $(A(\bar{x},\bar{u}), B(\bar{x}))$ and chosen positive-definite weight matrices $(Q, R)$, we solve the Riccati equation
\begin{equation}\label{eq:riccati}
\begin{aligned}
    &P(\bar{x},\bar{u})A(\bar{x},\bar{u}) + \tran{A(\bar{x},\bar{u})}P(\bar{x},\bar{u}) \\
    &\qquad- P(\bar{x},\bar{u})B(\bar{x})\inv{R}\tran{B(\bar{x})}P(\bar{x},\bar{u}) = -Q
\end{aligned}
\end{equation}
for the positive-definite solution $P(\bar{x},\bar{u})$. We then compute the tracking controller
\begin{equation}\label{eq:lqr-linear}
    u = \pi_{\mathrm{LQR}}(x, \bar{x}, \bar{u}) \defn \bar{u} - \inv{R}\tran{B(\bar{x})}P(\bar{x},\bar{u})e.
\end{equation}
In practice, the ``linearized LQR'' tracking controller (i.e., LQR control applied to the linearized dynamics) in \cref{eq:lqr-linear} can be effective as long as $(x(t),u(t))$ remains close to $(\bar{x}(t),\bar{u}(t))$, i.e., as long as the linearized error dynamics in \cref{eq:error-linearized} remain a good approximation of original error dynamics in \cref{eq:error}. Overall, the linearized LQR tracking controller requires us to be able to evaluate and differentiate the vector fields $(f, \{b_j\}_{j=1}^m)$; no additional structures are required.

\subsection{Nonlinear State-Dependent LQR}
For general nonlinear systems, the linearized LQR tracking controller presented in the previous section is a good first choice. However, it can fail for nonlinear systems when $(x(t),u(t))$ strays from the target $(\bar{x}(t),\bar{u}(t))$, since then \cref{eq:error-linearized} is no longer a good approximation.

In this section, we introduce an \emph{exact} nonlinear factorization of the error dynamics for general control-affine systems that resemble the linearized form in \cref{eq:error-linearized}. This factorization is based on the theory of SDC forms \cite{Cloutier1997,Cimen2010,Cimen2012}, and thereby enables a feedback law based on solving an associated SDRE.

\paragraph{State-Dependent LQR for Regulation} To begin, we first look at the simpler problem of regulating the state $x(t)$ of the system $\dot{x} = f(x) + B(x)u$ to $x = 0$. For now, we assume that $(x,u) = (0,0)$ is an equilibrium pair, i.e., $f(0) = 0$. If $f : \R^n \to \R^n$ is continuously differentiable, \citet[Proposition 1]{Cimen2010} shows
\begin{equation}\label{eq:sdc}
    \dot{x} = f(x) + B(x)u = A(x)x + B(x)u,
\end{equation}
where $f(x) \equiv A(x)x$ is an \emph{exact} factorization known as a \emph{state-dependent coefficient (SDC) form} of~$f$. With chosen positive-definite matrices $(Q, R)$, these factorized dynamics naturally enable the controller $u = K(x)x = -\inv{R}\tran{B(x)}P(x)x$, where $P(x)$ is the positive-definite solution of the \emph{state-dependent Riccati equation (SDRE)}
\begin{equation}\label{eq:sdre}
\begin{aligned}
    &P(x)A(x) + \tran{A(x)}P(x)& \\
    &\qquad - P(x)B(x)\inv{R}\tran{B(x)}P(x) = -Q
\end{aligned}.
\end{equation}
As its name implies, the SDRE is dependent on the \emph{current} state $x$ of the system. This contrasts with the Riccati equation for linearized LQR in \cref{eq:riccati}, which does not depend on $x$ and only depends on the target pair $(\bar{x},\bar{u})$ due to linearization. Despite using an exact nonlinear factorization of the dynamics, the feedback law $u = -\inv{R}\tran{B(x)}P(x)x$ is only locally stabilizing in theory and there is no guarantee that it will outperform linearized LQR, especially if only symmetric solutions of \cref{eq:sdre} are considered \citep[Example 2]{Cimen2012}. Nevertheless, \emph{state-dependent LQR (SD-LQR)} control in practice can induce a large region of attraction, especially relative to linearized control \citep{Cimen2012}.

\paragraph{Generalized SDC Forms} To extend SD-LQR to tracking control for control-affine systems, we leverage a generalization of SDC forms previously introduced by \citet{TsukamotoChungEtAl2021,TsukamotoChungEtAl2021b} and described in \cref{thm:sdc} below.
\begin{proposition}
    \label{thm:sdc}
    Suppose $f : \R^n \to \R^d$ is continuously differentiable. Then a matrix function $A : \R^n \x \R^n \to \R^{d \x n}$ exists such that
    \begin{equation}\label{eq:sdc-generalized}
        f(x) - f(\bar{x}) = A(\bar{x}, x - \bar{x})(x - \bar{x}) = A(\bar{x}, e)e,
    \end{equation}
    for all $x,\bar{x} \in \R^n$ with $e \defn x - \bar{x}$. Furthermore, $A$ can be chosen such that $A(\bar{x}, 0) \equiv \pd{f}{x}(\bar{x})$.
\end{proposition}\begin{proof}
    Consider any curve $r(s) = \bar{x} + R(s)e$ where $R : [0,1] \to \R^{n \x n}$ is differentiable, $R(0) = 0$, and $R(1) = I$. Then by the fundamental theorem for line integrals,
    \begin{equation}\label{eq:line-integral}
        f(x) - f(\bar{x})
        \equiv  \underbrace{\rbr*{
            \int_0^1 \pd{f}{x}(\bar{x} + R(s)e)R'(s)\,ds
        }}_{\eqqcolon A(\bar{x}, e)}e.
    \end{equation}
    Moreover, $A(\bar{x}, 0) = \int_0^1 \pd{f}{x}(\bar{x})R'(s)\,ds = \pd{f}{x}(\bar{x})$.
\end{proof}
\Cref{thm:sdc} describes a factorization of continuously differentiable $f$ that exactly quantifies $f(x) - f(\bar{x})$ between any $x$ and $\bar{x}$. When $x = \bar{x}$, the matrix factor $A(\bar{x}, e)$ reduces to the local Jacobian of $f$ at $\bar{x}$. Much like the linear approximation $\pd{f}{x}(\bar{x})e$, the exact factorization $A(\bar{x}, e)e$ is a function of the chosen ``target'' $\bar{x}$ and the ``error'' $e \defn x - \bar{x}$. It is precisely this perspective that now allows us to apply this generalized SDC form to tracking control.

\paragraph{State-Dependent LQR for Trajectory Tracking} For SD-LQR tracking control, we consider the error dynamics for general control-affine systems. Let $(\bar{x}(t), \bar{u}(t))$ be a dynamically feasible pair that we want to track. Then the dynamics of the tracking error $e \defn x - \bar{x}$ are
\begin{equation}\label{eq:error-nonlinear}
\hspace{-0.6em}\begin{aligned}
    \dot{e}
    &= f(x) + B(x)u - f(\bar{x}) - B(\bar{x})\bar{u} \\
    &= f(x) - f(\bar{x}) + (B(x){-}B(\bar{x}))\bar{u} + B(x)(u{-}\bar{u}) \\
    &= \underbrace{\biggl( A_0(\bar{x},e) + \sum_{j=1}^m\bar{u}_jA_j(\bar{x},e) \biggr)}_{
        \eqqcolon A_\text{SDC}(\bar{x},\bar{u},e)
    }e + B(x)v
\end{aligned}~,\end{equation}
where $v \defn u - \bar{u}$, and $(A_0, \{A_j\}_{j=1}^m)$ are SDC factorizations of the vector fields $(f, \{b_j\}_{j=1}^m)$ such that
\begin{equation}\label{eq:sdc-factorizations}
\begin{aligned}
    f(x) - f(\bar{x}) &\equiv A_0(\bar{x}, e)e \\
    b_j(x) - b_j(\bar{x}) &\equiv A_j(\bar{x}, e)e,\ \forall j \in \{1, 2, \dots, m\}
\end{aligned}~.\end{equation}
An SDRE similar to \cref{eq:sdre} expressed in terms of $(A_\text{SDC}(\bar{x},\bar{u},e), B(x))$ and chosen positive-definite weight matrices $(Q,R)$ can be solved for the positive-definite matrix $P_\text{SDC}(\bar{x},\bar{u},e)$. The associated nonlinear tracking controller is then
\begin{equation}\label{eq:lqr-sdc}
    u = \pi_\text{SDC}(x, \bar{x}, \bar{u}) \defn \bar{u} - \inv{R}\tran{B(x)}P_\text{SDC}(\bar{x},\bar{u},e)e.
\end{equation}
This controller reduces to the linearized LQR controller in \cref{eq:lqr-linear} if $A_\text{SDC}(\bar{x},\bar{u},0)$ is used, since then the SDC factorizations and hence the exact nonlinear error dynamics in \cref{eq:error-nonlinear} reduce to the Jacobians and the linearized error dynamics, respectively, in \cref{eq:error-linearized}.

Our goal in using SD-LQR tracking control is to enable better tracking performance for highly nonlinear systems that may experience large deviations from the target trajectory, e.g., during fast or aggressive maneuvers. The key trade-off in the use of a more complex controller is the need for additional known control-oriented structure. In this case, this structure comprises the SDC factorizations $(A_0, \{A_j\}_{j=1}^m)$ that are not required in the simpler linearized LQR tracking controller. In \cref{sec:learning-structure}, we will discuss how we can learn $(A_0, \{A_j\}_{j=1}^m)$ from data, and later in \cref{sec:experiments} we will show how this has a powerful regularization effect on learning models of dynamical systems for the purposes of closed-loop control. Before that, in the next section we overview alternative methods that couple a tracking controller with a certificate function guaranteeing closed-loop tracking convergence.

\subsection{Exponential Stabilizability via Contraction Theory}
Linearized and state-dependent LQR rely on approximate and exact factorized forms, respectively, of the system dynamics to construct tracking control laws. However, neither of these LQR controllers is guaranteed to stabilize the closed-loop error dynamics when the system is nonlinear. In this section, we review a family of tracking controllers that ensure exponential stability, i.e.,
\begin{equation}
    \norm{x(t) - \bar{x}(t)}_2 \leq \alpha\norm{x(0) - \bar{x}(0)}_2\exp(-\beta t),
\end{equation}
with overshoot $\alpha > 0$ and decay rate $\beta > 0$, for all $t \geq 0$.

To this end, \emph{contraction theory} \citep{LohmillerSlotine1998} seeks to construct certifiably stabilizing controllers for any control-affine system of the form \cref{eq:caf} by analyzing the stabilizability of the variational dynamics
\begin{equation}
    \dot{\delta}_x = \underbrace{\biggl(
        \pd{f}{x}(x) + \sum_{j=1}^m u_j \pd{b_j}{x}(x)
    \biggr)}_{\eqqcolon A(x,u)}\delta_x + B(x)\delta_u,
\end{equation}
where $\delta_x$ and $\delta_u$ are virtual displacements in the tangent spaces at $x$ and $u$, respectively. The high-level idea of contraction theory is to stabilize this infinite family of linear variational systems pointwise everywhere with a variational feedback law for~$\delta_u$, then path-integrate to get a stabilizing feedback law for~$u$ in the original system \citep{LohmillerSlotine1998,ManchesterSlotine2017}. Let $M : \R^n \to \mathbb{S}^n_{\succ 0}$ be a uniformly positive-definite matrix-valued function, i.e., such that $\underbar{\lambda}I \preceq M(x) \preceq \overbar{\lambda}I$ for some constants $\underbar{\lambda}, \overbar{\lambda} > 0$ and all $x \in \R^n$. Denote the time-derivative of $M(x)$ as $\dot{M}(x,u)$, with $ij$-th element
\begin{equation}
    \dot{M}_{ij}(x,u) \defn \tran{\grad{M_{ij}}(x)}(f(x) + B(x)u).
\end{equation}
Then $M(x)$ is a \emph{control contraction metric (CCM)} for the system in \cref{eq:caf} if there exist a constant $\beta > 0$ and a variational controller $\delta_u = \delta_\pi(\delta_x, x, u)$ such that
\begin{equation}\label{eq:ccm-variational}
\begin{aligned}
    &\tran{\delta_x}\rbr*{ \dot{M}(x,u){+}M(x)A(x,u){+}\tran{A(x,u)}M(x) }\delta_x \\
    &\quad + 2\tran{\delta_x}M(x)B(x)\delta_\pi(\delta_x, x, u) \leq -2\beta\tran{\delta_x}M(x)\delta_x
\end{aligned}\end{equation}
for all $\delta_x$, $x$, and $u$. Given a CCM, an exponentially stabilizing tracking controller of the form
\begin{equation}\label{eq:ccm-controller}
    u = \pi_\mathrm{CCM}(x,\bar{x},\bar{u}) = \bar{u} + k(x,\bar{x})
\end{equation}
can be constructed by geodesic integration between $x$ and~$\bar{x}$ \citep{ManchesterSlotine2017,SinghLandryEtAl2019,SinghRichardsEtAl2020}, with overshoot $\alpha = \sqrt{\overbar{\lambda}/\underbar{\lambda}}$, decay rate $\beta$, and ${k(\bar{x},\bar{x}) \equiv 0}$. Alternatively, a differentiable controller of the form in \cref{eq:ccm-controller} achieves this same result if
\begin{equation}\label{eq:closed-loop-contraction}
\hspace{-0.8em}
\begin{aligned}
    &\dot{M}(x,u) + \tran{ \rbr*{ A(x,u) + B(x)\pd{k}{x}(x,\bar{x}) } }M(x) \\
    &+ M(x)\rbr*{ A(x,u) + B(x)\pd{k}{x}(x,\bar{x}) } \preceq -2\beta M(x)
\end{aligned},
\end{equation}
for all $x$, $\bar{x}$, and $\bar{u}$ \citep{ManchesterSlotine2017}.

The exponential stability of the error dynamics in closed-loop with the tracking controller in \cref{eq:ccm-controller} is \emph{certified} by the CCM $M$. Once again we see that attaining better closed-loop performance requires additional control-oriented structure; in this case, this structure comprises the certificate $M$ and the closed-loop contraction condition in \cref{eq:closed-loop-contraction} that must be satisfied for \emph{all} $x$, $\bar{x}$, and $\bar{u}$.

\section{Jointly Learning Dynamics, Controllers, and Control-Oriented Structure}\label{sec:learning-structure}

In the previous section, we introduced a number of tracking controllers for nonlinear control-affine systems. We also highlighted how increasing the complexity of the tracking controller often promises improved closed-loop performance at the cost of requiring knowledge of additional control-oriented components. For linearized LQR, only the vector fields $(f, \{b_j\}_{j=1}^m)$ and their derivatives are needed. For SD-LQR, we also need to know the SDC factorizations $(A_0, \{A_j\}_{j=1}^m)$ of $(f, \{b_j\}_{j=1}^m)$. For CCM-based tracking control, we need to know $(f, B)$ and a CCM $M$ that together satisfy the constraint in \cref{eq:closed-loop-contraction} for all $x$, $\bar{x}$, and $\bar{u}$. Even when $(f, B)$ are known, synthesizing SDC factorizations (e.g., via the line integral in \cref{eq:line-integral}) or a CCM is a difficult problem that requires leveraging further structure in the dynamics (e.g., sparsity). This is generally not possible when $(f, B)$ are learned from data for an unknown system using complex parametric function approximators (e.g., neural networks).

In this section, we describe our main contribution to learning how to control control-affine dynamical systems when we only have access to a finite labelled data set $\mathcal{D} \defn \{(x^{(i)}, u^{(i)}, \dot{x}^{(i)})\}_{i=1}^N$ of input-output measurements of \cref{eq:caf}. Specifically, we describe a few methods for jointly learning a dynamics model and a tracking controller with unconstrained optimization, and focus on how this involves additionally modeling and learning control-oriented structure to enable a particular feedback law.

\paragraph{Learning Dynamics from Data} Each method in this section learns a model of the dynamics in \cref{eq:caf}. To this end, we define the regression loss
\begin{equation}\label{eq:regression}
    L^\text{dyn}_\text{reg}(f, B, \mathcal{D}) = \sum_{(x, u, \dot{x}) \in \mathcal{D}} \norm{\dot{x} - f(x) - B(x)u}_2^2.
\end{equation}
If we instantiate $(f,B)$ with parametric functions, such as neural networks, we can do gradient descent on this loss to fit $(f,B)$ to the data. Thus, a na\"{i}ve approach and our first baseline for learning how to control \cref{eq:caf} is to fit a differentiable model of $(f,B)$ to the data $\mathcal{D}$ and then apply linearized tracking LQR from \cref{sec:lqr-linear}.

\paragraph{Learning SDC Factorizations (Our Method)} For SD-LQR, we need to learn the SDC factorizations denoted by ${\mathcal{A} \defn (A_0, \{A_j\}_{j=1}^m)}$. For this, we use the regression loss
\begin{equation}\hspace{-0.9em}
    L^\text{SDC}_\text{reg}(\mathcal{A},\mathcal{D})
    = \hspace{-1em}\sum_{\substack{
        (x, u, \dot{x}),\\(\bar{x}, \bar{u}, \dot{\bar{x}}) \in \mathcal{D}
    }}\hspace{-0.8em}
    \norm{\dot{e} - A_\text{SDC}(\bar{x},\bar{u},e)e - B(x)v}_2^2,
\end{equation}
which sums over \emph{pairs} of labelled samples in the data set~$\mathcal{D}$. We also need $\mathcal{A}$ to be a set of valid SDC factorizations, for which we define the \emph{unlabelled} data set $\mathcal{D}^\text{SDC}_\text{aux} = \{(x^{(i)}, \bar{x}^{(i)})\}_{i=1}^{N^\text{SDC}_\text{aux}}$ and the auxiliary loss
\begin{equation}\hspace{-0.5em}
\begin{aligned}
    &L^\text{SDC}_\text{aux}(f,B,\mathcal{A},\mathcal{D}^\text{SDC}_\text{aux}) \\
    &= \hspace{-0.4em} \sum_{(x,\bar{x}) \in \mathcal{D}^\text{SDC}_\text{aux} }\hspace{-0.4em}
    \biggl(\begin{aligned}[t]
        &\norm{f(x) - f(\bar{x}) - A_0(\bar{x},e)e}_2^2 \\
        &+ \sum_{j=1}^m \norm{b_j(x) - b_j(\bar{x}) - A_j(\bar{x},e)e}_2^2
    \biggr)\end{aligned}
\end{aligned}.
\end{equation}
Overall, we can learn $(f,B,\mathcal{A})$ instantiated as parametric functions via gradient descent on the composite loss
\begin{equation}\label{eq:loss-sdc}
\begin{aligned}
    &L^\text{SDC}(f,B,\mathcal{A},\mathcal{D},\mathcal{D}^\text{SDC}_\text{aux}) \\
    &= L^\text{dyn}_\text{reg}(f, B, \mathcal{D}) + L^\text{SDC}_\text{reg}(\mathcal{A},\mathcal{D}) + L^\text{SDC}_\text{aux}(f,B,\mathcal{A},\mathcal{D}^\text{SDC}_\text{aux})
\end{aligned}.
\end{equation}
This total loss is \emph{semi-supervised} in that it is a function of both labelled and unlabelled data $\mathcal{D}$ and $\mathcal{D}^\text{SDC}_\text{aux}$, respectively. Ideally, we would want to constrain $\mathcal{A}$ to be a set of SDC factorizations of $(f,B)$ consistent with \cref{eq:sdc-generalized}. Since we cannot straightforwardly enforce \cref{eq:sdc-generalized} by construction, we use the auxiliary loss term in \cref{eq:loss-sdc} as a penalty-based relaxation, with as many unlabelled samples in $\mathcal{D}^\text{SDC}_\text{aux}$ as possible. This idea of relaxing pointwise functional constraints with sampling-based penalty terms is a common approach to learning global control-oriented structure \citep{RichardsBerkenkampEtAl2018,SinghRichardsEtAl2020,SunJhaEtAl2020,DawsonGaoEtAl2022} and more generally in semi-infinite optimization \citep{ZhangWuEtAl2010}.

\paragraph{Learning CCMs} This method is founded on the literature concerning joint learning of dynamics, controllers, and stability certificates \citep{SinghRichardsEtAl2020,SunJhaEtAl2020,DawsonGaoEtAl2022,ZhouQuartzEtAl2022}. For CCM-based tracking control, we need to learn a dynamics model $(f,B)$, a uniformly positive-definite CCM $M$, and a feedback controller ${u = \bar{u} + k(x,\bar{x})}$ such that $k(\bar{x},\bar{x}) \equiv 0$, that altogether satisfy the inequality in \cref{eq:closed-loop-contraction} for all $x$, $\bar{x}$, and $\bar{u}$. We take cues from \citet{SunJhaEtAl2020} to setup a loss function that will allow us to train all three components together with gradient descent, albeit with some adjustments to accommodate our lack of any knowledge of the dynamics $(f,B)$ (which \citet{SunJhaEtAl2020} assume are known).

We first specify the desired  overshoot $\alpha > 0$, decay rate $\beta > 0$, and eigenvalue lower bound $\underbar{\lambda} > 0$ as hyperparameters, and construct a candidate CCM~$M$ as
\begin{equation}
    M(x) = \underbar{\lambda}I + L(x)\tran{L(x)},
\end{equation}
where $L : \R^n \to \R^{n \x n}$ is any parametric matrix function. This construction ensures $M(x) \succeq \underbar{\lambda}I$ for all~$x$. To ensure $k(\bar{x},\bar{x}) \equiv 0$, we follow \cref{thm:sdc} and let $k(x,\bar{x}) = K(x,\bar{x})(x - \bar{x})$ for any parametric function ${K : \R^n \to \R^n \to \R^{m \x n}}$. With the closed-loop variational matrix defined by
\begin{equation}
    \bar{A}(x,\bar{x},\bar{u}) \defn \pd{f}{x}(x) + \sum_{j=1}^m u_j \pd{b_j}{x}(x) + B(x)\pd{k}{x}(x,\bar{x}),
\end{equation}
we collect terms of the inequality from \cref{eq:closed-loop-contraction} in
\begin{equation}\begin{aligned}
    C(x,\bar{x},\bar{u}) =\   &\dot{M}(x,u) + \tran{\bar{A}(x,\bar{x},\bar{u})}M(x) \\
                                    &+ M(x)\bar{A}(x,\bar{x},\bar{u}) + 2\beta M(x),
\end{aligned}\end{equation}
with $u = \bar{u} + k(x,\bar{x}) = \bar{u} + K(x,\bar{x})(x-\bar{x})$. Finally, with the unlabelled data set $\mathcal{D}^\text{CCM}_\text{aux} = \{(x^{(i)}, \bar{x}^{(i)}), \bar{u}^{(i)})\}_{i=1}^{N^\text{CCM}_\text{aux}}$, we define the auxiliary loss
\begin{equation}
\begin{aligned}
    &L^\text{CCM}_\text{aux}(f,B,M,K,\mathcal{D}^\text{CCM}_\text{aux}) \\
    &= \hspace{-1em}\sum_{(x,\bar{x},\bar{u}) \in \mathcal{D}^\text{CCM}_\text{aux} }\hspace{-0.4em} \begin{aligned}[t]\biggl(
    &\mathrm{max}\rbr*{ 0, \lambda_\mathrm{max}(C(x,\bar{x},\bar{u})) } \\
    &+ \mathrm{max}\rbr*{ 0, \lambda_\mathrm{max}(M(x)) - \alpha^2\underbar{\lambda} } \biggr)\end{aligned}
\end{aligned},
\end{equation}
where $\lambda_\mathrm{max}(\cdot)$ denotes the maximum eigenvalue operator. Overall, we can learn $(f,B,M,K)$ instantiated as parametric functions via gradient descent on the total loss
\begin{equation}\label{eq:loss-ccm}
\begin{aligned}
    &L^\text{CCM}(f,B,M,K,\mathcal{D},\mathcal{D}^\text{CCM}_\text{aux}) \\
    &= L^\text{dyn}_\text{reg}(f, B, \mathcal{D}) + L^\text{CCM}_\text{aux}(f,B,M,K,\mathcal{D}^\text{CCM}_\text{aux})
\end{aligned}.
\end{equation}
Much like in the SD-LQR case, this total loss is semi-supervised, although the auxiliary data set $\mathcal{D}^\text{CCM}_\text{aux}$ also requires samples of the input $\bar{u}$. This loss function can be viewed as an unconstrained relaxation of the approach from \citet{SinghRichardsEtAl2020}, who instead use pointwise inequalities derived from \cref{eq:ccm-variational} as exact constraints in an optimization over $(f,B,M)$. However, \citet{SinghRichardsEtAl2020} only use linear-in-parameter approximators for $(f,B,M)$ to construct a bi-convex program between $(f,B)$ and $M$, investigate the regularizing effect of fitting $(f,B,M)$ on the predictive capabilities of $(f,B)$ in closed-loop, and do not learn a controller. In contrast, the modified setup described above jointly learns a dynamics model, certificate function, and controller that can each be expressed with complex parametric functions, so that in the next section we can compare with the learning setups for linearized LQR and SD-LQR.

\section{Experiments}\label{sec:experiments}
\begin{figure*}[t]
\centering
\includegraphics[width=\textwidth]{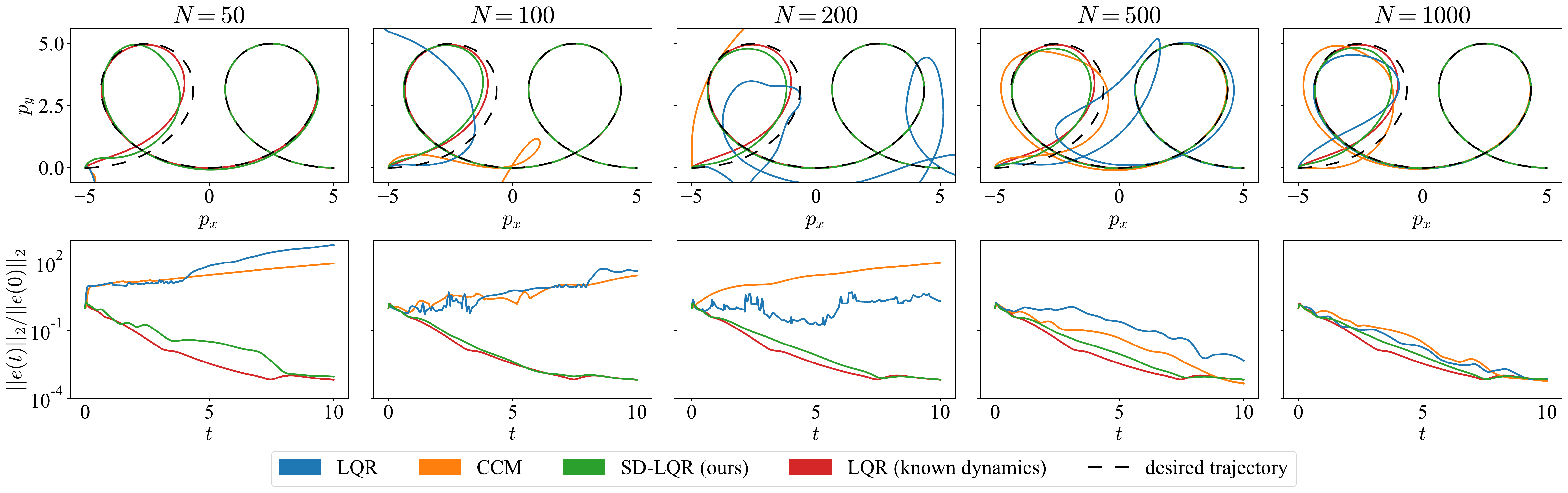}\vspace{-1em}
\caption{%
    Trajectory tracking results for the PVTOL system on a double loop-the-loop trajectory. The top row qualitative depicts the closed-loop trajectories for each method overlayed with the desired trajectory (black dashed). The bottom row shows the normalized tracking error over time. Plots proceed from left to right with an increasing amount $N$ of labelled training data. Our learned SD-LQR method is the only learning-based approach that successfully tracks the trajectory for all $N$.
}
\label{fig:double-loop}
\end{figure*}

In this section, we experimentally investigate the three methods described in \cref{sec:learning-structure} for jointly learning a dynamics model, stabilizing tracking controller, and/or some control-oriented structure enabling the controller, namely:
\begin{itemize}[nosep]
    \item\emph{Na\"ive LQR learning:} Fit a control-affine form $(f,B)$ to labelled data $\mathcal{D} \defn \{(x^{(i)}, u^{(i)}, \dot{x}^{(i)})\}_{i=1}^N$ via gradient descent on the regression loss in \cref{eq:regression}. Then perform linearized LQR.

    \item\emph{CCM learning:} Jointly fit $(f,B)$, a CCM $M$, and a gain matrix function $K$ to labelled data $\mathcal{D}$ and unlabelled data $\mathcal{D}^\text{CCM}_\text{aux} = \{(x^{(i)}, \bar{x}^{(i)}), \bar{u}^{(i)})\}_{i=1}^{N^\text{CCM}_\text{aux}}$ via gradient descent on the composite loss in \cref{eq:loss-ccm}. Then apply the controller $u = \bar{u} + K(x,\bar{x})(x - \bar{x})$.

    \item\emph{SD-LQR learning (our method):} Jointly fit $(f,B)$ and SDC factorizations $\mathcal{A} \defn (A_0, \{A_j\}_{j=1}^m)$ to labelled data~$\mathcal{D}$ and unlabelled data $\mathcal{D}^\text{SDC}_\text{aux} = \{(x^{(i)}, \bar{x}^{(i)})\}_{i=1}^{N^\text{SDC}_\text{aux}}$ via gradient descent on the composite loss in \cref{eq:loss-sdc}. Then perform SD-LQR.
\end{itemize}
We also implement linearized LQR with \emph{known} dynamics as an oracle. We evaluate these methods on two nonlinear benchmark systems:

\paragraph{Spacecraft} Our planar spacecraft, based on that of \citet{LewEtAl2022}, has mass $m$ with center-of-mass offset at $(d_x, d_y) \in \R^2$, and a rotational moment of inertia $J$. Its state is $x = (p_x, p_y, \theta, \dot{p}_x, \dot{p}_y, \dot{\theta}) \in \R^6$, where $(p_x,p_y)$ is its position and $\theta$ is its heading angle. The control is $u = (F_x, F_y, M) \in \R^3$, where $(F_x, F_y)$ are the applied forces along the inertial $x$-axis and $y$-axis, respectively, and $M$ is the applied moment. The control-affine dynamics of the spacecraft are given by
\begin{equation*}
    f(x) = \frac{1}{m}\bmx{\dot{p}_x \\ \dot{p}_y \\ \dot{\theta} \\ \dot{\theta}^2d_x \\ \dot{\theta}^2d_y \\ 0},\
    B(x) = \frac{1}{mJ}\bmx{
        0 & 0 & 0 \\
        0 & 0 & 0 \\
        0 & 0 & 0 \\
        J + d_y^2   & -d_xd_y   & d_y  \\
        -d_xd_y     & J + d_x^2 & -d_x \\
        md_y        & -md_x     & m
    }.
\end{equation*}

\paragraph{PVTOL} Our planar vertical-take-off-and-landing (PVTOL) vehicle has mass~$m$, rotational moment of inertia~$J$, moment arm length~$\ell$ between the center of mass and each of two rotors, and gravitational acceleration $g$. Its state is ${x = (p_x, p_y, \phi, v_x, v_y, \dot{\phi}) \in \R^6}$, where $(p_x,p_y)$ is its position, $\phi$ is its roll angle, and $(v_x,v_y)$ is its velocity in the body-fixed frame. The control is $u = (F_R, F_L) \in \R^2$, where $F_R$ and $F_L$ are the applied thrusts by the right and left rotors, respectively, along the body-fixed $y$-axis. The control-affine dynamics of this PVTOL are given by
\begin{equation*}
    f(x) = \bmx{
        v_x\cos\phi - v_y\sin\phi \\
        v_x\sin\phi + v_y\cos\phi \\
        \dot\phi \\
        v_y\dot\phi - g\sin\phi \\
        -v_x\dot\phi - g\cos\phi \\
        0
    },\
    B(x) = \bmx{
        0 & 0 \\
        0 & 0 \\
        0 & 0 \\
        0 & 0 \\
        \sfrac{1}{m} & \sfrac{1}{m} \\
        \sfrac{\ell}{J} & -\sfrac{\ell}{J}
    }.
\end{equation*}
The planar spacecraft is only slightly nonlinear due to the term $\dot{\theta}^2$ introduced by the center-of-mass offset, and so should serve as a relatively easy benchmark for learning-based control. In contrast, the PVTOL is a highly nonlinear, underactuated, non-minimum-phase dynamical system \citep{HauserSastryEtAl1992}, and thus serves as a challenging benchmark.

\paragraph{Training Details} For each system, we begin by uniformly sampling points $\{(x^{(i)}, u^{(i)})\}_{i=1}^N$ from a bounded state-control set $\mathcal{X} \x \mathcal{U} \subset \R^n \x \R^m$, and evaluating the true dynamics to form the labelled data $\mathcal{D}$. Both $\mathcal{X}$ and $\mathcal{U}$ are described in \cref{appendix}, along with other implementation details and hyperparameters. We additionally uniformly sample unlabelled data sets $\mathcal{D}^\text{CCM}_\text{aux}$ and $\mathcal{D}^\text{SDC}_\text{aux}$ for use with the CCM and SDC learning methods, respectively, from $\mathcal{X}$ and $\mathcal{U}$. We vary the labelled training set size $N$ to investigate the data efficiency of each method, with a constant number of auxiliary points $N^\text{CCM}_\text{aux} = N^\text{SDC}_\text{aux} = 10000$. Each function in $(f, B, M, K, A_0, \{A_j\}_{j=1}^m)$ is approximated as a feedforward neural network with the same number of fully connected hidden layers, and appropriately shaped input and output dimensions using Python and JAX \citep{BradburyFrostigEtAl2018}. For each method, the appropriate subset of these functions is trained via the Adam optimizer \citep{KingmaBa2015} on the corresponding loss function. Training is performed for $50000$ epochs while the loss on a held-out validation set is monitored; for each method, the model parameters corresponding to the lowest validation loss are chosen for testing. This training procedure is repeated for each method across 5~random seeds.

\begin{figure*}[t]
\centering
\includegraphics[width=\textwidth]{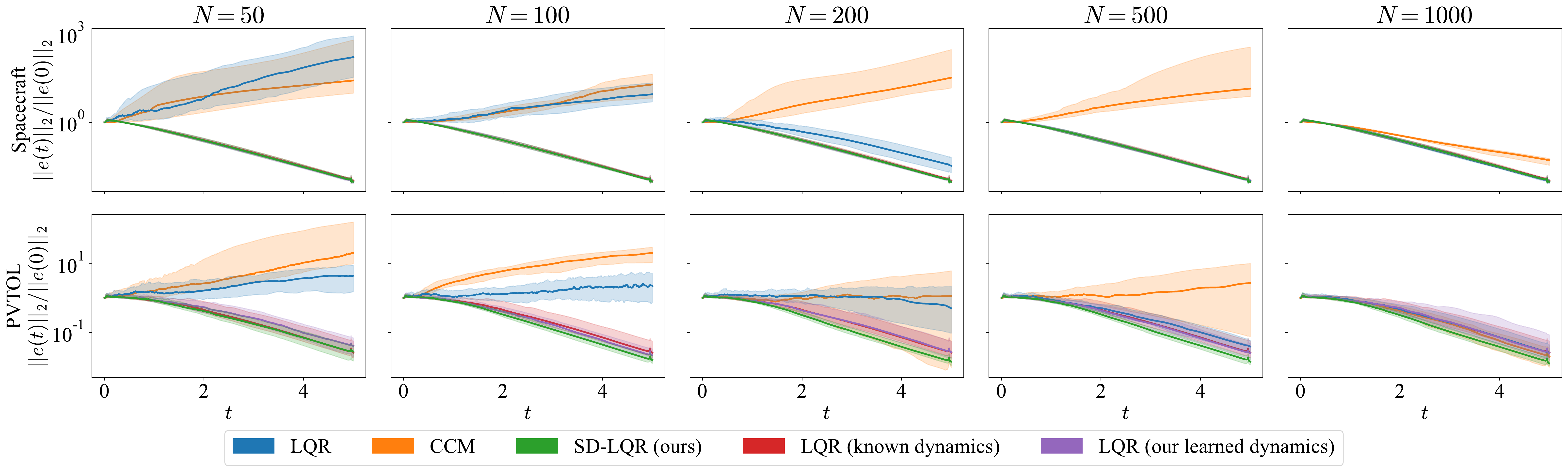}\vspace{-1em}
\caption{%
    Trajectory tracking results for both the spacecraft and PVTOL systems for $N_\text{test} = 100$ trajectories each. The top and bottom rows show the normalized tracking error over time for the spacecraft and PVTOL, respectively. Plots proceed from left to right with an increasing amount $N$ of labelled training data. Colored lines represent the median across all trajectories at each time $t$, while shaded regions depict interquartile ranges. Our learned SD-LQR method consistently outperforms the considered baseline learning methods.
}\label{fig:trends}
\end{figure*}

\paragraph{Testing and Results} To test the controllers learned with each method, we must first generate dynamically feasible trajectories for tracking. We first evaluate the PVTOL system qualitatively; we leverage its differential flatness \citep{Ailon2010} to generate a feasible pair $(\bar{x}(t), \bar{u}(t))$ yielding the double loop-the-loop shape in \cref{fig:double-loop}. For a single random seed, we plot the closed-loop trajectory from using each learned controller to track the loop-the-loop. We repeat this test for various sizes $N$ of the labelled training data set $\mathcal{D}$, and plot the trajectories in $(p_x,p_y)$-space and the normalized tracking error~$\frac{\norm{e(t)}_2}{\norm{e(0)}_2}$ over time. Our learned SD-LQR method is the only learning-based method that succeeds for every size~$N$, while the learned LQR and CCM controllers outright fail for smaller data set sizes. This is initial evidence of the \emph{data efficiency} in learning SDC factorizations for the purpose of control.

For more thorough testing, we want to generate many trajectories in a manner applicable to both the spacecraft and PVTOL. To this end, we generate $N_\text{test} = 100$ feasible trajectories $\mathcal{T}_\text{test} \defn \{(\bar{x}^{(k)}(t), \bar{u}^{(k)}(t))\}_{k=1}^{N_\text{test}}$ for each system by solving the optimal control problem
\begin{equation}\label{eq:ocp}
\begin{aligned}
    &\minimize_{\bar{x}(\cdot),\bar{u}(\cdot)}~
    \int_0^T \rbr*{\tran{\bar{x}(t)}Q\bar{x}(t) + \tran{\bar{u}(t)}R\bar{u}(t)}\,dt \\
    &\subjectto~\begin{aligned}[t]
        &\dot{\bar{x}}(t) = f(\bar{x}(t)) + B(\bar{x}(t))\bar{u}(t) \\
        &\bar{x}(0) = \bar{x}^{(k)}_0 \\
        &\bar{x}(T) = 0 \\
        &u_\mathrm{lb} \preceq \bar{u}(t) \preceq u_\mathrm{ub}
    \end{aligned}
\end{aligned}\end{equation}
for different initial conditions $\bar{x}_0^{(k)}$ sampled uniformly from~$\mathcal{X}$, where $(Q,R)$ are positive-definite weight matrices and $(u_\mathrm{lb}, u_\mathrm{ub})$ are control input bounds. Specifically, we use CasADi \citep{AnderssonGillisEtAl2019} to transcribe this problem into a nonlinear multiple shooting optimization that is passed to and solved by the Ipopt solver \citep{WachterBiegler2006}. Then, for each system, test trajectory, and tracking controller, we uniformly sample an initial state $x^{(k)}_0 \neq \bar{x}^{(k)}_0$ from $\mathcal{X}$, and simulate the closed-loop system.

\cref{fig:trends} displays the normalized tracking error $\frac{\norm{e(t)}_2}{\norm{e(0)}_2}$ over time for both the spacecraft and the PVTOL, for various training set sizes $N$. For each method, system, and $N$, the median normalized tracking error across the test trajectory set $\mathcal{T}_\text{test}$ is plotted along with shaded regions denoting the interquartile range over time. Once again we observe the data efficiency of our learned SD-LQR method for both systems. Our method even \emph{outperforms the oracle} LQR controller at higher values of $N$ for the PVTOL, despite having to learn the dynamics. This is likely due to how even the oracle LQR controller is limited by its linear approximation of the error dynamics, while the SD-LQR controller uses a learned model of the full nonlinear error dynamics. The na\"ive learned LQR method unsurprisingly converges to performance similar to the oracle LQR controller as $N$ increases. Notably, the CCM-based controller has the most trouble overall, thereby highlighting its data inefficiency and brittleness.

As an ablation in \cref{fig:trends}, with our dynamics model $(f,B)$ learned alongside the SDC factorizations $\mathcal{A}$ via the loss \cref{eq:loss-sdc}, we also perform linearized LQR rather than SD-LQR. We see that even solely using our learned dynamics model greatly improves the performance of linearized LQR across all values of $N$. This indicates that learning SDC factorizations as a structural prior \emph{regularizes} the dynamics model for closed-loop control. Nevertheless, if they are available, using the SDC factorizations directly in SD-LQR is more practical than repeatedly differentiating the dynamics online for linearized LQR.

\begin{figure*}[t]
\centering
\includegraphics[width=\textwidth]{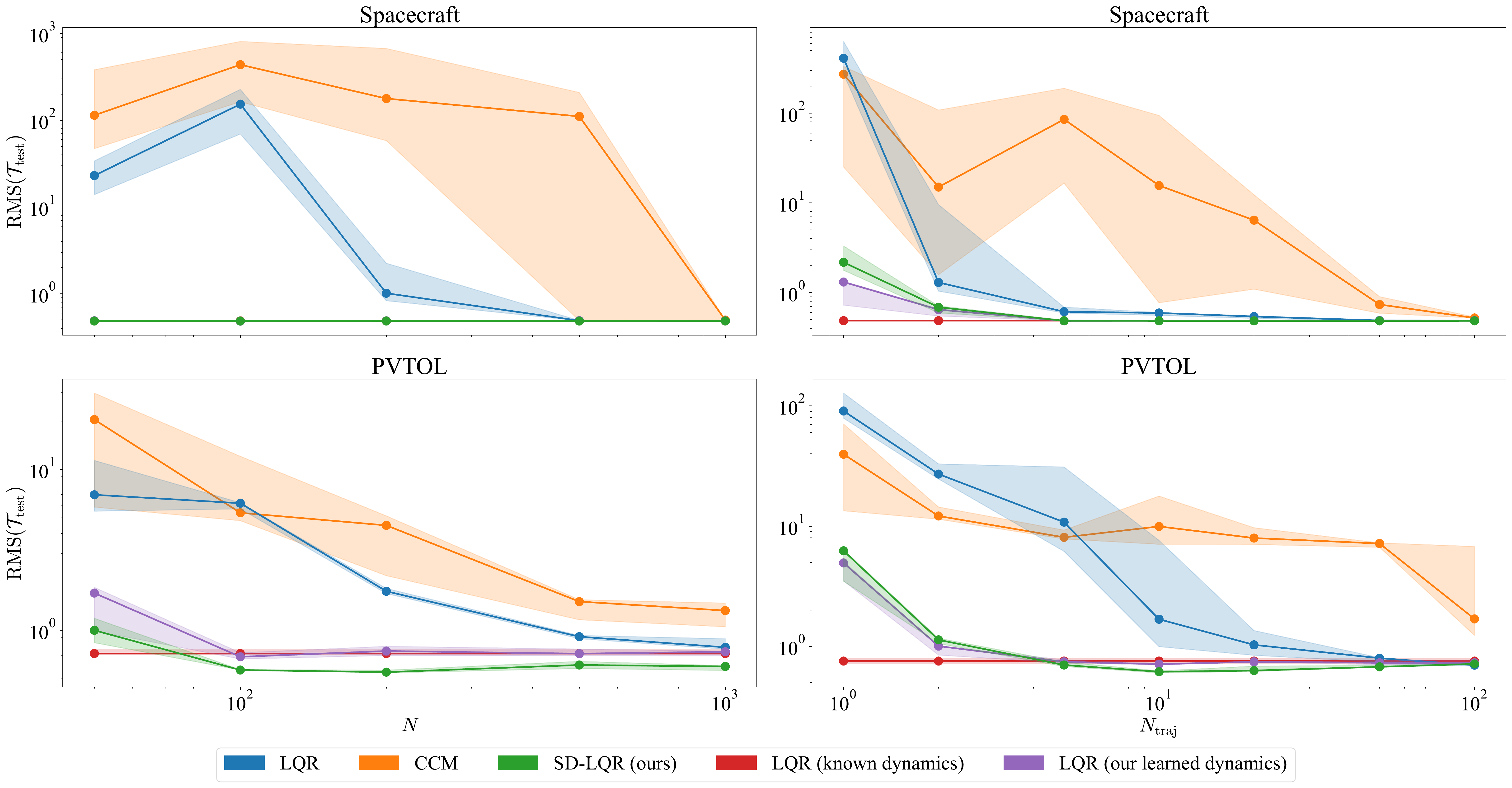}\vspace{-1em}
\caption{%
    RMS tracking error as a function of the labelled training data set size $N$ (left) and $N_\text{traj}$ (right), averaged across $N_\text{test} = 100$ test trajectories (see \cref{eq:rms}). Colored lines denote medians across $5$~random seeds, while shaded regions depict interquartile ranges. When training is done on uniformly sampled data, our learned SD-LQR method outperforms all other methods, even the oracle LQR on the PVTOL system. When training is done on trajectory data, performance is still at its best when using our dynamics model $(f,B)$ learned alongside SDC factorizations~$\mathcal{A} \defn (A_0, \{A_j\}_{j=1}^m)$. For $N_\text{traj} \in \{1,2\}$, linearized LQR with our learned model seems to do better than SD-LQR, although the performance difference is small.
}
\label{fig:efficiency}
\end{figure*}
To complete our assessment, we repeat the training procedure and tests underlying \cref{fig:trends} with $5$~random seeds, and aggregate the results in \cref{fig:efficiency}. Moreover, we repeat this entire process with \emph{trajectory} data. Specifically, instead of uniformly sampling points $\{(x^{(i)}, u^{(i)})\}_{i=1}^N$,  we solve the optimal control problem \cref{eq:ocp} to generate $N_\text{traj}$ dynamically feasible trajectories, different from those in the test set and each beginning at a different initial condition. We then collect data ``on-policy'' by simulating linearized LQR with the true dynamics to track these trajectories and generate a dataset. Each trajectory is $5$~seconds long with samples collected at $100$~Hz, such that $N = 500 N_\text{traj}$. However, samples along a trajectory are clearly temporally correlated, so we use $N_\text{traj}$ rather than~$N$ as a measure of the size of the dataset when training on trajectory data. For both uniformly sampled data and trajectory data, we consider the average root mean squared (RMS) error
\begin{equation}\label{eq:rms}
    \mathrm{RMS}(\mathcal{T}_\text{test}) \defn \frac{1}{N_\text{test}}
    \sum_{k=1}^{N_\text{test}}\sqrt{
        \frac{1}{T}\int_0^T \frac{\norm{e^{(k)}(t)}_2^2}{\norm{e^{(k)}(0)}_2^2} \,dt
    }
\end{equation}
across all test trajectories for each random seed and training set size $N$ (or $N_\text{traj}$). In \cref{fig:efficiency}, we plot the median and interquartile range of $\mathrm{RMS}(\mathcal{T}_\text{test})$ across random seeds as a function of $N$ and $N_\text{traj}$. From these plots, we can see an even starker contrast between the performance of our learned SD-LQR method and the others. For the spacecraft, our method is close in performance to that of the oracle LQR, which is not surprising given that the spacecraft dynamics are only mildly nonlinear. For the highly nonlinear PVTOL, our method begins outperforming the oracle LQR at only $N = 100$ and $N_\text{traj} = 5$. Meanwhile, both the learned linearized LQR and CCM controllers struggle until more training data is used, which highlights their data inefficiency compared to our method. Moreover, in the ablation wherein we apply linearized LQR to our dynamics model $(f,B)$ learned alongside SDC factorizations~$\mathcal{A}$, much of the performance gained does not require using the SDC factorizations directly. Once again, this indicates the usefulness of SDC factorizations as a \emph{regularizing} structural prior when they are learned alongside the model $(f,B)$.

\section{Conclusions and Future Work}\label{sec:conclusion}

In this paper, we studied how to jointly learn a dynamics model and a stabilizing tracking controller from only a finite data set of input-output measurements of an unknown dynamical system. We highlighted the importance of not only learning the dynamics, but also control-oriented structure that enables performant controller design. For this purpose, we proposed learning SDC factorizations of the dynamics for the purpose of state-dependent LQR tracking controller. Inspired by the literature, we compared our method to na\"ively learning a model for linearized LQR, and to methods that couple learned controllers with learned certificate functions. Overall, we found that our method outperformed the baselines in terms of data efficiency and tracking capability. Moreover, we observed that including SDC factorizations in the learning problem regularizes the dynamics model to perform better during closed-loop LQR-based control.

\paragraph{Future Work} We view this paper in part as a critique of methods that try to enforce closed-loop stabilizability guarantees by penalizing sampled violations of certificate conditions like \cref{eq:closed-loop-contraction}. As we have demonstrated, such methods can be data inefficient and brittle in learning good controllers, although the performance guarantees they are meant to certify (e.g., exponential stability) are attractive. Unlike these methods, our method learns \emph{intrinsic} structure in the dynamics to enable control, rather than simultaneously learning a parametric controller. Thus, an interesting avenue for future research lies in building system models that are \emph{intrinsically stabilizable}. This could build off of existing work in parameterizing dynamics models in part by stability certificates such that they are stable by construction \citep{ManekKolter2019,RevayWangEtAl2021}, albeit for the controlled case.

\section*{Acknowledgements}
We thank Masha Itkina for her invaluable feedback. We also thank the reviewers for their thoughtful input. This research was supported in part by the National Aeronautics and Space Administration (NASA) University Leadership Initiative via grant \#80NSSC20M0163. Spencer M.~Richards was also supported in part by the Natural Sciences and Engineering Research Council of Canada (NSERC). This article solely reflects our own opinions and conclusions, and not those of any NASA or NSERC entity.

\bibliography{ASL_papers,main}

\newcommand{\noopsort}[1]{} \newcommand{\printfirst}[2]{#1}
  \newcommand{\singleletter}[1]{#1} \newcommand{\switchargs}[2]{#2#1}
\begin{thebibliography}{42}
\providecommand{\natexlab}[1]{#1}
\providecommand{\url}[1]{\texttt{#1}}
\expandafter\ifx\csname urlstyle\endcsname\relax
  \providecommand{\doi}[1]{doi: #1}\else
  \providecommand{\doi}{doi: \begingroup \urlstyle{rm}\Url}\fi

\bibitem[Abate et~al.(2021)Abate, Ahmed, Giacobbe, and
  Peruffo]{AbateAhmedEtAl2021}
Abate, A., Ahmed, D., Giacobbe, M., and Peruffo, A.
\newblock Formal synthesis of {Lyapunov} neural networks.
\newblock \emph{{IEEE Control Systems Letters}}, 5\penalty0 (3):\penalty0
  773--778, 2021.
\newblock \doi{10.1109/LCSYS.2020.3005328}.

\bibitem[Acarman(2009)]{Acarman2009}
Acarman, T.
\newblock Nonlinear optimal integrated vehicle control using individual braking
  torque and steering angle with on-line control allocation by using
  state-dependent {Riccati} equation technique.
\newblock \emph{{Vehicle System Dynamics}}, 47\penalty0 (2):\penalty0 157--177,
  2009.
\newblock \doi{10.1080/00423110801932670}.

\bibitem[Ailon(2010)]{Ailon2010}
Ailon, A.
\newblock Simple tracking controllers for autonomous {VTOL} aircraft with
  bounded inputs.
\newblock \emph{{IEEE Transactions on Automatic Control}}, 55\penalty0
  (3):\penalty0 737--743, 2010.

\bibitem[Andersson et~al.(2019)Andersson, Gillis, Horn, Rawlings, and
  Diehl]{AnderssonGillisEtAl2019}
Andersson, J. A.~E., Gillis, J., Horn, G., Rawlings, J.~B., and Diehl, M.
\newblock {CasADi}: A software framework for nonlinear optimization and optimal
  control.
\newblock \emph{{Mathematical Programming Computation}}, 11\penalty0
  (1):\penalty0 1--36, 2019.

\bibitem[Banks et~al.(2002)Banks, Beeler, Kepler, and
  Tran]{BanksBeelerEtAl2002}
Banks, H.~T., Beeler, S.~C., Kepler, G.~M., and Tran, H.~T.
\newblock Reduced order modeling and control of thin film growth in an {HPCVD}
  reactor.
\newblock \emph{{SIAM Journal on Applied Mathematics}}, 62\penalty0
  (4):\penalty0 1251--1280, 2002.

\bibitem[Boffi et~al.(2020)Boffi, Tu, Matni, Slotine, and
  Sindhwani]{BoffiTuEtAl2020}
Boffi, N.~M., Tu, S., Matni, N., Slotine, J.-J.~E., and Sindhwani, V.
\newblock Learning stability certificates from data.
\newblock In \emph{{Conf.\ on Robot Learning}}, 2020.

\bibitem[Bradbury et~al.(2018)Bradbury, Frostig, Hawkins, Johnson, Leary,
  Maclaurin, Necula, Paszke, Vander{P}las, Wanderman-{M}ilne, and
  Zhang]{BradburyFrostigEtAl2018}
Bradbury, J., Frostig, R., Hawkins, P., Johnson, M.~J., Leary, C., Maclaurin,
  D., Necula, G., Paszke, A., Vander{P}las, J., Wanderman-{M}ilne, S., and
  Zhang, Q.
\newblock {JAX}: {Composable} transformations of {Python}+{NumPy} programs,
  2018.
\newblock {Available at }\url{http://github.com/google/jax}.

\bibitem[\c{C}imen(2010)]{Cimen2010}
\c{C}imen, T.
\newblock Systematic and effective design of nonlinear feedback controllers via
  the state-dependent {Riccati} equation ({SDRE}) method.
\newblock \emph{{Annual Reviews in Control}}, 34\penalty0 (1):\penalty0 32--51,
  2010.
\newblock \doi{10.1016/j.arcontrol.2010.03.001}.

\bibitem[\c{C}imen(2012)]{Cimen2012}
\c{C}imen, T.
\newblock Survey of state-dependent {Riccati} equation in nonlinear optimal
  feedback control synthesis.
\newblock \emph{{AIAA Journal of Guidance, Control, and Dynamics}}, 35\penalty0
  (4):\penalty0 1025--1047, 2012.
\newblock \doi{10.2514/1.55821}.

\bibitem[Chang \& Gao(2021)Chang and Gao]{ChangGao2011}
Chang, Y.-C. and Gao, S.
\newblock Stabilizing neural control using self-learned almost {Lyapunov}
  critics.
\newblock In \emph{{Proc.\ IEEE Conf.\ on Robotics and Automation}}, 2021.

\bibitem[Chang et~al.(2019)Chang, Roohi, and Gao]{ChangRoohiEtAl2019}
Chang, Y.-C., Roohi, N., and Gao, S.
\newblock Neural {Lyapunov} control.
\newblock In \emph{{Conf.\ on Neural Information Processing Systems}}, 2019.

\bibitem[Cloutier(1997)]{Cloutier1997}
Cloutier, J.~R.
\newblock State-dependent {Riccati} equation techniques: An overview.
\newblock In \emph{{American Control Conference}}, 1997.

\bibitem[Cloutier \& Zipfel(1999)Cloutier and Zipfel]{CloutierZipfel1999}
Cloutier, J.~R. and Zipfel, P.~H.
\newblock Hypersonic guidance via the state-dependent {Riccati} equation
  control method.
\newblock In \emph{{IEEE Conf.\ on Control Applications}}, 1999.

\bibitem[Dawson et~al.(2021)Dawson, Qin, Gao, and Fan]{DawsonQinEtAl2021}
Dawson, C., Qin, Z., Gao, S., and Fan, C.
\newblock Safe nonlinear control using robust neural {Lyapunov}-barrier
  functions.
\newblock In \emph{{Conf.\ on Robot Learning}}, 2021.

\bibitem[Dawson et~al.(2022)Dawson, Gao, and Fan]{DawsonGaoEtAl2022}
Dawson, C., Gao, S., and Fan, C.
\newblock Safe control with learned certificates: A survey of neural
  {Lyapunov}, barrier, and contraction methods.
\newblock Available at \url{https://arxiv.org/abs/2202.11762}, 2022.

\bibitem[Djeumou et~al.(2022)Djeumou, Neary, Goubault, Putot, and
  Topcu]{DjeumouNearyEtAl2022}
Djeumou, F., Neary, C., Goubault, E., Putot, S., and Topcu, U.
\newblock Neural networks with physics-informed architectures and constraints
  for dynamical systems modeling.
\newblock In \emph{{Learning for Dynamics \& Control}}, 2022.

\bibitem[Gupta et~al.(2020)Gupta, Menda, Manchester, and
  Kochenderfer]{GuptaMendaEtAl2020}
Gupta, J.~K., Menda, K., Manchester, Z., and Kochenderfer, M.~J.
\newblock Structured mechanical models for robot learning and control.
\newblock In \emph{{Learning for Dynamics \& Control}}, 2020.

\bibitem[Hauser et~al.(1992)Hauser, Sastry, and Meyer]{HauserSastryEtAl1992}
Hauser, J., Sastry, S., and Meyer, G.
\newblock Nonlinear control design for slightly non-minimum phase systems:
  Application to {V/STOL} aircraft.
\newblock \emph{{Automatica}}, 28\penalty0 (4):\penalty0 665--679, 1992.

\bibitem[Kingma \& Ba(2015)Kingma and Ba]{KingmaBa2015}
Kingma, D.~P. and Ba, J.~L.
\newblock Adam: A method for stochastic optimization.
\newblock In \emph{{Int.\ Conf.\ on Learning Representations}}, 2015.

\bibitem[Lew et~al.(2022)Lew, Sharma, Harrison, Bylard, and
  Pavone]{LewEtAl2022}
Lew, T., Sharma, A., Harrison, J., Bylard, A., and Pavone, M.
\newblock Safe active dynamics learning and control: A sequential
  exploration-exploitation framework.
\newblock \emph{{IEEE Transactions on Robotics}}, 38\penalty0 (5):\penalty0
  2888--2907, 2022.

\bibitem[Li et~al.(2022)Li, Duong, and Atanasov]{LiDuongEtAl2022}
Li, Z., Duong, T., and Atanasov, N.
\newblock Safe autonomous navigation for systems with learned {SE(3)}
  {Hamiltonian} dynamics.
\newblock In \emph{{Learning for Dynamics \& Control}}, 2022.

\bibitem[Lohmiller \& Slotine(1998)Lohmiller and Slotine]{LohmillerSlotine1998}
Lohmiller, W. and Slotine, J.-J.~E.
\newblock On contraction analysis for non-linear systems.
\newblock \emph{{Automatica}}, 34\penalty0 (6):\penalty0 683--696, 1998.

\bibitem[Manchester \& Slotine(2017)Manchester and
  Slotine]{ManchesterSlotine2017}
Manchester, I.~R. and Slotine, J.-J.~E.
\newblock Control contraction metrics: Convex and intrinsic criteria for
  nonlinear feedback design.
\newblock \emph{{IEEE Transactions on Automatic Control}}, 62\penalty0
  (6):\penalty0 3046--3053, 2017.

\bibitem[Manek \& Kolter(2019)Manek and Kolter]{ManekKolter2019}
Manek, G. and Kolter, J.~Z.
\newblock Learning stable deep dynamics models.
\newblock In \emph{{Conf.\ on Neural Information Processing Systems}}, 2019.

\bibitem[Murray et~al.(1994)Murray, Li, and Sastry]{MurrayLiEtAl1994}
Murray, R.~M., Li, Z., and Sastry, S.~S.
\newblock \emph{A Mathematical Introduction to Robotic Manipulation}.
\newblock {CRC Press}, 1 edition, 1994.

\bibitem[Revay et~al.(2021)Revay, Wang, and Manchester]{RevayWangEtAl2021}
Revay, M., Wang, R., and Manchester, I.~R.
\newblock Recurrent equilibrium networks: Unconstrained learning of stable and
  robust dynamical models.
\newblock In \emph{{Proc.\ IEEE Conf.\ on Decision and Control}}, 2021.
\newblock \doi{10.1109/CDC45484.2021.9683054}.

\bibitem[Richards et~al.(2018)Richards, Berkenkamp, and
  Krause]{RichardsBerkenkampEtAl2018}
Richards, S.~M., Berkenkamp, F., and Krause, A.
\newblock The {Lyapunov} neural network: Adaptive stability certification for
  safe learning of dynamical systems.
\newblock In \emph{{Conf.\ on Robot Learning}}, 2018.

\bibitem[Richards et~al.(2021)Richards, Azizan, Slotine, and
  Pavone]{RichardsAzizanEtAl2021}
Richards, S.~M., Azizan, N., Slotine, J.-J., and Pavone, M.
\newblock Adaptive-control-oriented meta-learning for nonlinear systems.
\newblock In \emph{{Robotics: Science and Systems}}, 2021.

\bibitem[Richards et~al.(2023)Richards, Azizan, Slotine, and
  Pavone]{RichardsAzizanEtAl2023}
Richards, S.~M., Azizan, N., Slotine, J.-J., and Pavone, M.
\newblock Control-oriented meta-learning.
\newblock \emph{{Int.\ Journal of Robotics Research}}, 2023.
\newblock In press.

\bibitem[Sindhwani et~al.(2018)Sindhwani, Tu, and
  Khansari]{SindhwaniTuEtAl2018}
Sindhwani, V., Tu, S., and Khansari, M.
\newblock Learning contracting vector fields for stable imitation learning.
\newblock {Available at }\url{https://arxiv.org/abs/1804.04878}, 2018.

\bibitem[Singh et~al.(2019)Singh, Landry, Majumdar, Slotine, and
  Pavone]{SinghLandryEtAl2019}
Singh, S., Landry, B., Majumdar, A., Slotine, J.-J.~E., and Pavone, M.
\newblock Robust feedback motion planning via contraction theory.
\newblock \emph{{Int.\ Journal of Robotics Research}}, 2019.
\newblock Submitted.

\bibitem[Singh et~al.(2021)Singh, Richards, Sindhwani, Slotine, and
  Pavone]{SinghRichardsEtAl2020}
Singh, S., Richards, S.~M., Sindhwani, V., Slotine, J.-J.~E., and Pavone, M.
\newblock Learning stabilizable nonlinear dynamics with contraction-based
  regularization.
\newblock \emph{{Int.\ Journal of Robotics Research}}, 40\penalty0
  (10--11):\penalty0 1123--1150, 2021.

\bibitem[Slotine \& Li(1987)Slotine and Li]{SlotineLi1987}
Slotine, J.-J.~E. and Li, W.
\newblock On the adaptive control of robot manipulators.
\newblock \emph{{Int.\ Journal of Robotics Research}}, 6\penalty0 (3):\penalty0
  49--59, 1987.

\bibitem[Slotine \& Li(1991)Slotine and Li]{SlotineLi1991}
Slotine, J.-J.~E. and Li, W.
\newblock \emph{Applied Nonlinear Control}.
\newblock {Prentice Hall}, 1991.

\bibitem[Sun et~al.(2020)Sun, Jha, and Fan]{SunJhaEtAl2020}
Sun, D., Jha, S., and Fan, C.
\newblock Learning certified control using contraction metric.
\newblock In \emph{{Conf.\ on Robot Learning}}, 2020.

\bibitem[Tsukamoto et~al.(2021{\natexlab{a}})Tsukamoto, Chung, and
  Slotine]{TsukamotoChungEtAl2021}
Tsukamoto, H., Chung, S.-J., and Slotine, J.-J.~E.
\newblock Neural stochastic contraction metrics for learning-based control and
  estimation.
\newblock \emph{{IEEE Control Systems Letters}}, 5\penalty0 (5):\penalty0
  1825--1830, 2021{\natexlab{a}}.

\bibitem[Tsukamoto et~al.(2021{\natexlab{b}})Tsukamoto, Chung, and
  Slotine]{TsukamotoChungEtAl2021b}
Tsukamoto, H., Chung, S.-J., and Slotine, J.-J.~E.
\newblock Contraction theory for nonlinear stability analysis and
  learning-based control: A tutorial overview.
\newblock \emph{{Annual Reviews in Control}}, 52:\penalty0 135--169,
  2021{\natexlab{b}}.

\bibitem[W\"{a}chter \& Biegler(2006)W\"{a}chter and
  Biegler]{WachterBiegler2006}
W\"{a}chter, A. and Biegler, L.~T.
\newblock On the implementation of an interior-point filter line-search
  algorithm for large-scale nonlinear programming.
\newblock \emph{{Mathematical Programming}}, 106:\penalty0 25--57, 2006.

\bibitem[Watanabe et~al.(2008)Watanabe, Iwase, Hatakeyama, and
  Maruyama]{WatanabeIwaseEtAl2008}
Watanabe, K., Iwase, M., Hatakeyama, S., and Maruyama, T.
\newblock Control strategy for a snake-like robot based on constraint force and
  verification by experiment.
\newblock In \emph{{IEEE/RSJ Int.\ Conf.\ on Intelligent Robots \& Systems}},
  2008.

\bibitem[Zhang et~al.(2010)Zhang, Wu, and L{\'o}pez]{ZhangWuEtAl2010}
Zhang, L., Wu, S.-Y., and L{\'o}pez, M.~A.
\newblock A new exchange method for convex semi-infinite programming.
\newblock \emph{{SIAM Journal on Optimization}}, 20\penalty0 (6):\penalty0
  2959--2977, 2010.

\bibitem[Zhong et~al.(2020)Zhong, Dey, and Chakraborty]{ZhongDeyEtAl2020}
Zhong, Y.~D., Dey, B., and Chakraborty, A.
\newblock Symplectic {ODE-Net}: Learning {Hamiltonian} dynamics with control.
\newblock In \emph{{Int.\ Conf.\ on Learning Representations}}, 2020.

\bibitem[Zhou et~al.(2022)Zhou, Quartz, De~Sterck, and Liu]{ZhouQuartzEtAl2022}
Zhou, R., Quartz, T., De~Sterck, H., and Liu, J.
\newblock Neural {Lyapunov} control of unknown nonlinear systems with stability
  guarantees.
\newblock In \emph{{Conf.\ on Neural Information Processing Systems}}, 2022.

\end{thebibliography}
\bibliographystyle{icml2023}

\newpage
\appendix
\onecolumn
\section{Hyperparameters and Implementation Details}\label{appendix}

\paragraph{Physical Parameters} For the spacecraft, we set its mass to $m=0.5$, rotational moment of inertia to $J = 0.005$, and its center-of-mass offset to $(d_x,d_y) = (0.1,0.1)$. For the PVTOL, we set the its mass to $m=0.5$, arm length to $\ell = 0.25$, rotational moment of inertia to $J = 0.005$, and gravitational acceleration to $g = 9.81$.

\paragraph{Hyperparameters} Each function in $(f, B, M, K, A_0, \{A_j\}_{j=1}^m)$ is approximated as a feedforward neural network with two hidden layers and $128$ hidden $\mathrm{tanh}$ activation units per layer. We use the Adam optimizer \citep{KingmaBa2015} with a learning rate of $10^{-3}$ and otherwise default hyperparameters. Training is performed for $50000$ epochs while the loss on a held-out validation set of size $0.10 N$ is monitored, where $N$ is the size of the labelled training data set. For each method, the model parameters corresponding to the lowest validation loss are chosen for testing.

For the CCM-based learning method, since \cref{eq:closed-loop-contraction} is homogeneous in $M(x)$, we choose $\underbar{\lambda} = 0.1$ without loss of generality. Additionally, we fix the overshoot $\alpha = 10$ and the decay rate $\beta = 0.5$ in the auxiliary loss \cref{eq:loss-ccm}. For both the CCM and SDC learning methods, we use $N^\text{CCM}_\text{aux} = N^\text{SDC}_\text{aux} = 10000$ unlabelled samples.

\paragraph{Sampling} For sampling states and inputs, we draw uniformly from bounded sets $\mathcal{X} \subset \R^n$ and $\mathcal{U} \subset \R^m$, respectively. For the spacecraft, we use
\begin{equation}\begin{aligned}
    \mathcal{X}
    &= \{x \in \R^6 \mid -c \preceq x \preceq c,\ c \defn (1, 1, \pi, 0.2, 0.2, 0.25)\}
    \\
    \mathcal{U}
    &= \{u \in \R^3 \mid -c \preceq u \preceq c,\ c \defn (1, 1, 0.1)\}
\end{aligned}.\end{equation}
For the PVTOL, we use
\begin{equation}\begin{aligned}
    \mathcal{X}
    &= \{x \in \R^6 \mid -c \preceq x \preceq c,\ c \defn (10, 10, \pi/3, 2, 1, \pi/3)\}
    \\
    \mathcal{U}
    &= \{u \in \R^2 \mid (0.1mg, 0.1mg) \preceq u \preceq (2mg, 2mg)\}
\end{aligned},\end{equation}
where $m$ and $g$ are the vehicle mass and gravitational acceleration, respectively. We also use $\mathcal{U}$ to define the control bounds in the optimal control problem for generating test trajectories.

\paragraph{Testing} When generating test trajectories with the optimal control problem in \cref{sec:experiments}, we use $Q = I$ and $R = 0.01I$ in the cost function for both systems. For simulating the linearized and state-dependent LQR tracking controllers, we use $Q = I$ and $R = I$ in their corresponding Riccati equations.

\end{document}